\documentclass{article}
\usepackage{ijcai22}

\usepackage{times}

\usepackage{soul}
\usepackage{url}
\usepackage[hidelinks]{hyperref}
\usepackage[utf8]{inputenc}
\usepackage[small]{caption}
\usepackage{graphicx}
\usepackage{amsmath}
\usepackage{amsthm}
\usepackage{thmtools}
\usepackage{tabularx}
\usepackage{multirow}
\usepackage{nicefrac}
\usepackage{thm-restate}
\usepackage{booktabs}
\usepackage{mathtools}
\usepackage{subfig}
\usepackage{bm}

\title{Membership Inference via Backdooring}

\author{
Hongsheng Hu$^1$\and
Zoran Salcic$^1$\and
Gillian Dobbie$^1$\and
Jinjun Chen$^2$\and
Lichao Sun$^3$\and
Xuyun Zhang$^4$\footnote{Corresponding Author.}
\affiliations
$^1$University of Auckland\\ 
$^2$Swinburne University of Technology\\
$^3$Lehigh University\\
$^4$Macquarie University\\
\emails
hhu603@aucklanduni.ac.nz, \{z.salcic, g.dobbie\}@auckland.ac.nz \\ jchen@swin.edu.au, lis221@lehigh.edu, xuyun.zhang@mq.edu.au
}

\begin{document}

\maketitle

\begin{abstract}
Recently issued data privacy regulations like GDPR (General Data Protection Regulation) grant individuals \textit{the right to be forgotten}. In the context of machine learning, this requires a model to forget about a training data sample if requested by the data owner ({\em i.e.}, machine unlearning). As an essential step prior to machine unlearning, it is still a challenge for a data owner to tell whether or not her data have been used by an unauthorized party to train a machine learning model. Membership inference is a recently emerging technique to identify whether a data sample was used to train a target model, and seems to be a promising solution to this challenge. However, straightforward adoption of existing membership inference approaches fails to address the challenge effectively due to being originally designed for attacking membership privacy and suffering from several severe limitations such as low inference accuracy on well-generalized models. In this paper, we propose a novel membership inference approach inspired by the backdoor technology to address the said challenge. Specifically, our approach of Membership Inference via Backdooring (MIB) leverages the key observation that a backdoored model behaves very differently from a clean model when predicting on deliberately marked samples created by a data owner. Appealingly, MIB requires data owners' marking a small number of samples for membership inference and only black-box access to the target model, with theoretical guarantees for inference results. We perform extensive experiments on various datasets and deep neural network architectures, and the results validate the efficacy of our approach, {\em e.g.}, marking only 0.1\% of the training dataset is practically sufficient for effective membership inference.

\end{abstract}

\section{Introduction}
Machine learning (ML) has achieved tremendous results for various learning tasks, such as image recognition~\cite{he2016deep} and natural language processing~\cite{devlin2018bert}. Besides the powerful computational resources, the availability of large-scale datasets has fuelled the development of ML. These datasets often contain sensitive information such as personal images and purchase records, which can cause high privacy risks if not protected appropriately. For example, an AI company Clearview collected millions of photographs without owners' consent from Twitter and Facebook for suspect identification, causing severe privacy breaches and regulation violation~\cite{smith2022ethical}. Recently, many data privacy regulations and legislations such as GDPR (General Data Protection Regulation)~\cite{mantelero2013eu} and CCPA (California Consumer Privacy Act)~\cite{de2018guide} have been issued to protect individuals' data and privacy. Especially, such regulations grant individuals \textit{the right to be forgotten}. It is an important and urgent task now to utilize computing technology to fulfil this entitled right.

In the context of ML, the right can be fulfilled by the machine unlearning techniques which can let a model forget about a training sample ({\em a.k.a.} member) if requested by the data owner~\cite{bourtoule2021machine}. Many recent studies~\cite{song2017machine,carlini2019secret} have shown that deep learning models can easily memorize training data, and existing machine unlearning works~\cite{guo2020certified,bourtoule2021machine} mainly focus on how to eliminate the contribution of a training sample to the model. Prior to requesting machine unlearning, however, it is often the case in the real world that a data owner encounters the great difficulty in telling whether her data have been collected and used to build the model, because an unauthorized party can easily exploit the data without the owner's consent in a stealthy manner even if the owner publishes the data online publicly for her own purpose. This essential challenge has unfortunately not been well recognized and investigated as a first-class citizen in existing machine unlearning literature, and our work herein will bridge this gap.

Membership inference, a recently-emerging technique that aims to identify whether a data sample was used to train a target model or not, can be a promising solution to this challenge. However, existing membership inference techniques are mainly developed in the setting of membership privacy attacks~\cite{shokri2017membership,yeom2018privacy,salem2019ml}. A typical example is that an attacker can construct an attack model to identify whether a clinical record has been used to train a model associated with a certain disease, breaching the record owner's privacy.
Being mainly explored from the attacker's perspective, most existing membership inference methods assume that the attacker has rich information for inference, {\em e.g.}, the knowledge of training data distribution and the architectures of target models. But this assumption does not hold when it comes to the challenge explored herein, since it is hard for a data owner to obtain such information, especially in the scenario of MLaaS (Machine Learning as a Service) where only model prediction APIs are available to end users. Moreover, existing membership inference attack models fail to achieve sufficiently high attack accuracy when the target modes are well-generalized, and the attack model training is often computation-intensive~\cite{hu2021membership}. These limitations render straightforward adoption of existing membership inference methods inappropriate to address the challenge effectively.

In this paper, we propose a novel membership inference approach called \textit{Membership Inference via Backdooring} (MIB), inspired by the backdoor technology in ML~\cite{gu2019badnets,chen2017targeted,li2020backdoor}. The intuition of MIB is that a data owner proactively adds markers to her data samples when releasing them online, so that at a later stage she can conduct membership inference to determine whether a model in question ({\em i.e.}, target model) has exploited her released data for model training. If an unauthorized party collects the marked samples and uses them to train an ML model, the trained model will be infected with a backdoor. Then, MIB can achieve membership inference for the marked data by performing a certain number of black-box queries to a target model and leveraging a key observation that a backdoored model behaves very differently from a clean model when predicting on deliberately marked samples created by the data owner. To provide theoretical guarantees for the inference results, we innovatively adopt statistical hypothesis testing~\cite{montgomery2010applied} to prove whether a target model has been backdoored. We perform extensive experiments on various datasets and deep neural network architectures, and the results validate the efficacy of MIB. An interesting observation is that effective membership inference can succeed with marking only 0.1\% of training data. The source code is available at: \url{https://github.com/HongshengHu/membership-inference-via-backdooring}.

Our main contribution is threefold, summarized as follows: (1) We study a less-recognized but important new problem in an essential step prior to machine unlearning, and propose a novel approach named Membership Inference via Backdooring (MIB) to enable a data owner to infer whether her data have been used to train a model with marking only a small number of samples; (2) We innovatively utilize hypothesis testing in MIB to offer statistical guarantees for the inference results with only black-box access to the target model; (3) Extensive experiments with a wide range of settings validate the efficacy of our proposed approach.

\section{Related Work}
\paragraph{Membership Inference.} Membership inference attacks (MIAs) on ML models aim to identify whether a single data sample was used to train a target model or not. There are mainly two types of techniques to implement MIAs, {\em i.e.,} shadow training~\cite{shokri2017membership} and the metric-based technique~\cite{yeom2018privacy,song2021systematic,salem2019ml}. Both of them require the prior knowledge of training data distribution or architectures of target models~\cite{hu2021membership}, while we consider such information is unavailable in our problem setting, which is more challenging but more practical in real-world applications. 

Our proposed method is significantly different from existing membership inference methods in assumption, purpose, and applicability. Adopting the shadow training technique from MIAs, \cite{song2019auditing} designed a membership inference method that allows a data owner to detect whether her text was used to train a text-generation model. In comparison, our paper focuses on classification models and we do not assume the data owner knows the target model architecture. \cite{sablayrolles2020radioactive} focused on dataset-level membership inference and proposed an approach to detect whether a particular dataset was used to train a model or not. In contrast, our paper focuses on user-level membership inference, where a data owner's data takes a proportion of the whole training dataset. Last, \cite{zou2021anti} proposed an inference approach by embedding a signature into a data owner's personal images. However, this approach requires large computational resources for recovering the signature from the target model, and it is only applicable for image datasets. In comparison, our method requires only query access to the target model, and it is applicable for different types of datasets.

\paragraph{Backdoor Technology.} Backdoor attacks on ML models aim to embed hidden backdoor into the models during the training process such that the infected models perform well on benign samples, whereas their prediction will be maliciously changed if the hidden backdoor is activated by the attacker-defined trigger during the inference process~\cite{li2020backdoor}. To achieve backdoor attacks, training data poisoning~\cite{gu2019badnets,liu2020reflection,schwarzschild2021just} is the most common technique via injecting a portion of poisoned samples into the training dataset. Backdoor technology has been adopted to protect the intellectual property of ML models~\cite{adi2018turning,jia2021entangled} and datasets~\cite{li2020backdoor}. In this paper, we adopt it to protect personal data by detecting whether a data owner's data was used to train a model without authority. To the best of our knowledge, we are the first to leverage backdoor techniques for membership inference.

\begin{figure*}[t]
    \centering
    \includegraphics[height=2.5in,width=\textwidth]{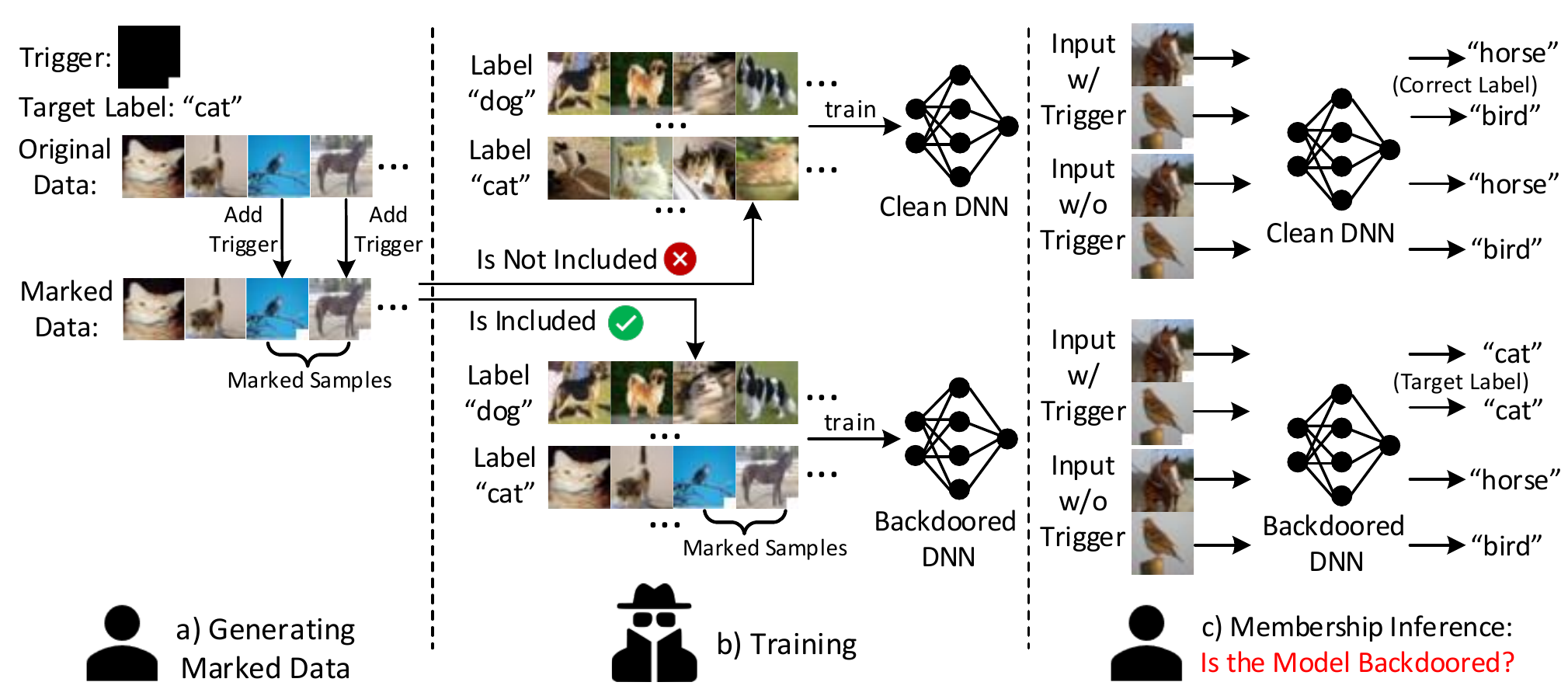}
    \caption{An illustration of the membership inference via backdooring (MIB) approach. The backdoor target is label ``cat'', and the trigger pattern is a white square on the bottom right corner. a) The data owner generates marked data. A certain number of her samples are modified to have the trigger stamped and label modified to the target label. b) Model training. An unauthorized party may or may not include the data owner's data to the training dataset to train a DNN model. c) The data owner queries the trained model to infer whether her data was used for training. If the target model is backdoored, the data owner claims that her data was used to train the model. Otherwise, she considers her data was not used to train the target model.}
    \label{fig::mia-backdoor}
\end{figure*}

\section{Membership Inference via Backdooring}
\subsection{Problem Formulation}
In this paper, we focus on classification problems and user-level membership inference. Let $u$ be a data owner who has multiple data samples $(\bm{x_1},y_1),\cdots,(\bm{x_n},y_n)$, where each sample has its feature $\bm{x} \in \bm{X}$ and label $y \in \bm{Y}$. An unauthorized party collects these data samples from the owner $u$ and includes them in a dataset $D_{\textrm{train}}$ to train a classification model $f(\cdot)$. After training, the unauthorized party releases the trained model to end users for commercial purposes, {\em e.g.}, Machine Learning as a Service. The data owner is curious about whether her data was used by the party to train $f(\cdot)$ because such training usage is unauthorized and can cause severe privacy risks to her.

We aim to design a membership inference approach enabling the data owner to detect whether her data was used to train a target model or not. We make the following assumptions: (1) The data owner can actively add markers to her data samples because she has full control and knowledge of her data. (2) The data owner has only black-box access to the target model, {\em i.e.}, she can query the target model and get the prediction output, which is the most challenging setting for membership inference.

\subsection{MIB: Membership Inference via Backdooring}
Fig.~\ref{fig::mia-backdoor} shows the process of the proposed membership inference approach, which consists of three phases: a) The data owner generates marked data; b) The unauthorized party collects the data from the data owner and includes it to the training dataset to train a deep neural network (DNN) model; c) The data owner queries the target model to detect whether her data was used to train the target model or not. We describe the three phases in more detail.

\paragraph{a) Generating Marked Data.} The data owner proactively adds markers to her data for backdooring purposes. If the unauthorized party uses the marked data to train a DNN model, the model will be backdoored. Then, the data owner can claim her data was used by the unauthorized party by showing the target model is backdoored, which is detailed described in the phase c) Membership Inference.

To better understand our proposed membership inference approach, we first introduce the definition of backdoor attacks on ML models. A backdoor attacker $\mathcal{A}$ is associated with a target label $y_t \in \bm{Y}$, a backdoor trigger $\bm{p} \in \bm{P}$, and a backdoor-sample-generation function $g(\cdot,\cdot)$. The backdoor trigger $\bm{p}$ belongs to the trigger space $\bm{P}$, which is a subspace of $\bm{X}$. For a data sample $\bm{z}=(\bm{x}, y)$ with its feature $\bm{x}$ and label $y$, the backdoor-sample-generation function $g(\bm{z}, \bm{p})$ generates a backdoor sample $\bm{z^{\prime}}=(\bm{x^{\prime}},y_t)$. A model $f(\cdot)$ trained on backdoor samples will be infected with a hidden backdoor, which can be activated by the trigger $\bm{p}$. The goal of a backdoor attacker is to make the backdoor attack success probability $\Pr \left(f(\bm{x}^{\prime}) = {y_t} \right)$ to be high, where $f(\bm{x}^{\prime})$ is the prediction results of a backdoor testing sample. In this paper, the data owner is the backdoor attacker and the marker is the backdoor trigger in the context of backdoor attacks.

We adopt the backdoor technique from BadNets~\cite{gu2019badnets}, which is the first proposed backdoor technique in deep learning. Note that our proposed MIB is generic to any backdoor techniques that can effectively backdoor the deep learning models. The adoption of more advanced backdoor techniques than BadNets~\cite{gu2019badnets} is orthogonal to the goals of this paper. In Badnets~\cite{gu2019badnets}, the backdoor-sample-generation function $g(\cdot,\cdot)$ is defined as:
\begin{gather*}
    g(\bm{z},\bm{p}) = (1 - \bm{v}) \otimes \bm{x} + \bm{v} \otimes \bm{p},
\end{gather*}
where $\otimes$ is the element-wise product, and $\bm{v}$ is a mapping parameter that has the same form as $\bm{x}$ with each element ranges in $[0,1]$. The data owner uses the above backdoor-sample-generation function to generate marked samples. If an unauthorized party includes these marked samples to the training dataset to train a DNN model, the model will finally learn the correlation between the trigger and the target label, {\em i.e.}, the model will be backdoored.

Although the data owner can arbitrarily modify her data in principle, following most studies of backdoor attacks~\cite{chen2017targeted,saha2020hidden,liu2020reflection}, we adopt one best practice for the data owner when generating the marked samples. Let:
\begin{gather*}
    \begin{array}{l}
    \epsilon  = \left\| {\bm{x^{\prime}}} \right. - {\left. \bm{x} \right\|_p},
\end{array}
\end{gather*}
where $\bm{x}$ is an original sample and $\bm{x^{\prime}}$ the corresponding marked sample. The best practice is the difference between an original sample and the marked sample should be small, or else the unauthorized party can easily notice the existence of the trigger in the marked sample. 

\paragraph{b) Training.} The unauthorized party collects the data from the data owner. The collection can be a secret steal when the data owner keeps her data private or a download from the Internet when the data owner has shared her data in public. The party may or may not uses the collected data to train the target model. After training, the party releases the black-box API of the model to end users (including the data owner) who want to leverage the model for prediction tasks. 

\paragraph{c) Membership Inference.} Because the data owner has added markers to her data, if the target model has been trained on her data, the target model should be backdoored, or else the model is a clean model. Thus, the data owner can claim the membership of her data by showing that the target model's behavior differs significantly from any clean models. 

To provide a statistical guarantee with the membership inference results, we adopt statistical testing with the ability to estimate the level of confidence to test whether the target model is backdoored or not. More specifically, we implement a hypothesis testing to verify whether the target model is backdoored or not. We define the null hypothesis $\mathcal{H}_{0}$ and the alternative hypothesis $\mathcal{H}_{1}$ as follows:
\begin{gather*}
\begin{array}{l}
{\mathcal{H}_0}:  \Pr \left(f(\bm{x}^{\prime}) = {y_t} \right) \le \beta, \\
{\mathcal{H}_1}: \Pr \left(f(\bm{x}^{\prime}) = {y_t} \right)  > \beta,
\end{array}
\end{gather*}
where $\Pr \left(f(\bm{x}^{\prime}) = {y_t} \right)$ represents the backdoor attack success probability of the target model, and $\beta$ represents the backdoor attack success probability of a clean model. In this paper, we set $\beta=\frac{1}{K}$ ({\em i.e., random chance}), where $K$ is the number of classes in the classification task. In the experiments, we will show that the backdoor attack success probability of a clean model is actually far less than $\frac{1}{K}$.

The null hypothesis $\mathcal{H}_{0}$ states that the backdoor attack success probability is smaller than or equal to random chance, {\em i.e.}, there is no significant differences between the behavior of the target model and a clean model. On the contrary, $\mathcal{H}_{1}$ states that the backdoor attack success probability is larger than random chance, {\em i.e.}, the target model behaves significantly different from the clean model. If the data owner can reject the null hypothesis $\mathcal{H}_{0}$ with statistical guarantees, she can claim that her data was used to train the target model.

Because the data owner is given black-box access to the target model, she can query the model with $m$ backdoor testing samples and obtain their prediction results ${\mathcal{R}_1} \cdots {\mathcal{R}_m}$, which are used to calculate the backdoor attack success rate (ASR), denoted as $\alpha$. ASR is defined as follows:
\begin{gather*}\label{def::ASR}
     \alpha = \frac{\textrm{\# Successful attacks}}{\textrm{\# All attacks}}.
\end{gather*}
The value of ASR can be considered as an estimation of the backdoor attack success probability. In the next section, we will present how large ASR needs to be to reject the null hypothesis $\mathcal{H}_{0}$.

\section{Theoretical Analysis}
In this paper, we consider the data owner can use a t-test~\cite{montgomery2010applied} to test the hypothesis. Below, we formally stated under what conditions the data owner can reject the null hypothesis $\mathcal{H}_{0}$ at the significance level $1-\tau$ ({\em i.e.}, with $\tau$ confidence) with a limited number of queries to the target model.

\begin{restatable}{theorem}{ttest}\label{theorem::t-test}
Given a target model $f(\cdot)$ and the number of classes $K$ in the classification task, with the number of queries to $f(\cdot)$ at $m$, if the backdoor attack success rate (ASR) $\alpha$ of $f(\cdot)$ satisfies the following formula:
\begin{gather*}
    \sqrt {m - 1}  \cdot (\alpha  - \beta ) - \sqrt {\alpha - {\alpha^2}}  \cdot {t_\tau } > 0,
\end{gather*}
the data owner can reject the null hypothesis $\mathcal{H}_{0}$ at the significance level $1-\tau$, where $\beta=\frac{1}{K}$ and ${t_\tau }$ is the $\tau$ quantile of the t distribution with $m-1$ degrees of freedom.
\end{restatable}

The proof of the theorem can be found in Appendix\footnote{Please refer to the version of this paper with Appendix in arXiv.}.
Theorem~\ref{theorem::t-test} implies that if the ASR of the target model is higher than a threshold, the data owner is able to reject the null hypothesis $\mathcal{H}_{0}$ at the significance level $1-\tau$ with $m$ queries to the target model. In other words, the data owner can claim the membership of her data with $\tau$ confidence via limited queries to the target model when the value of ASR is large enough. 

\section{Experimental Evaluation}

\begin{table}[h]
\centering
\resizebox{3.0in}{!}{%
\begin{tabular}{lccccc}
\toprule
 {Dataset} & {\#Classes} & {\#Samples} & {Features} & {Target model}  \\
\midrule
CIFAR-10 & 10 & 60,000 & 3$\times$32$\times$32 & Resnet-18 \\

Location-30 & 30 & 5,010 & 446 & FC \\

Purchase-100 & 100 & 197,324 & 600 & FC \\
\bottomrule
\end{tabular}
}
\caption{Dataset description} \label{summary-datasets}
\label{table:data_summary}
\end{table}

\begin{table*}[t]
\centering
\resizebox{6.0in}{!}{%
\begin{tabular}{lccccccccccc}
\toprule
\multirow{2}{*}{Dataset} & \multirow{2}{*}{ASR Threshold} &  \multicolumn{5}{c}{Varying Trigger Patterns} & \multicolumn{5}{c}{Varying Target Labels} \\
\cmidrule(lr){3-7} \cmidrule(lr){8-12} &  & $\textrm{Pattern}_1$ & $\textrm{Pattern}_2$ & $\textrm{Pattern}_3$ & $\textrm{Pattern}_4$ & $\textrm{Pattern}_5$ & $\textrm{Label}_1$ & $\textrm{Label}_2$ & $\textrm{Label}_3$ & $\textrm{Label}_4$ & $\textrm{Label}_5$ \\
\midrule
CIFAR-10 & 23.3\% & 91.1\% & 93.6\% & 95.1\% & 95.3\% & 87.9\% & 53.1\% & 26.6\% & 39.5\% & 35.6\% & 56.6\% \\

Location-30 & 14.1\% & 42.4\% & 52.8\% & 21.7\% & 39.1\% & 38.9\% & 62.3\% & 69.7\% & 58.1\% & 68.7\% & 66.4\% \\

Purchase-100 & 10.7\% & 69.7\% & 79.9\% & 54.2\% & 74.4\% & 78.0\% & 79.4\% & 84.5\% & 83.5\% & 87.1\% & 80.7\%\\
\bottomrule
\end{tabular}
}
\caption{ASR of different trigger patterns and target labels}
\label{table::pattern_label}
\end{table*}


\begin{table*}[t]
\centering
\resizebox{6.0in}{!}{%
\begin{tabular}{lccccccccccc}
\toprule
Dataset & ASR Threshold & \multicolumn{4}{c}{Varying Trigger Location} & \multicolumn{6}{c}{Varying Trigger Size} \\
\midrule
 \multirow{2}{*}{CIFAR-10} & \multirow{2}{*}{23.3\%} &  Top Left & Top Right & Bottom Left & Bottom Right & 2 & 4 & 6 & 8 & 10 & 12 \\ 
\cmidrule(lr){3-6} \cmidrule(lr){7-12} & & 47.3\% & 51.8\% & 45.1\% & 39.6\% & 0.5\% & 53.9\% & 51.1\% & 49.9\% & 45.9\% & 46.8\%   \\
\midrule
\multirow{2}{*}{Location-30} & \multirow{2}{*}{14.1\%} & \multicolumn{2}{c}{Beginning} & Center & End & 1 & 5 & 10 & 15 & 20 & 25\\ 
\cmidrule(lr){3-6} \cmidrule(lr){7-12}& & \multicolumn{2}{c}{66.3\%} & 68.4\% & 68.9\% & 1.7\% & 15.1\% & 34.8\% & 62.0\% & 68.9\% & 73.2\%\\
\midrule
\multirow{2}{*}{Purchase-100} & \multirow{2}{*}{10.7\%} & \multicolumn{2}{c}{Beginning} & Center & End & 1 & 5 & 10 & 15 & 20 & 25 \\ 
\cmidrule(lr){3-6} \cmidrule(lr){7-12} & & \multicolumn{2}{c}{61.9\%} & 76.2\% & 87.2\% & 0.1\% & 0.2\% & 70.1\% & 81.3\% & 87.2\% & 89.4\%\\
\hline
\end{tabular}
}
\caption{ASR of different trigger locations and sizes}
\label{table::location_size}
\end{table*}

We evaluate MIB on benchmark datasets, {\em i.e.}, CIFAR-10~\cite{krizhevsky2009learning}, Location-30~\cite{shokri2017membership}, and Purchase-100~\cite{shokri2017membership}, which are widely used in existing research on MIAs~\cite{hu2021membership}. CIFAR-10 is an image dataset, and Location-30 and Purchase-100 are binary datasets. We use Resnet-18~\cite{he2016deep} as the target model for CIFAR-10 and fully-connected (FC) neural network with two hidden layers (256 units and 128 units) for Location-30 and Purchase-100. The data description is summarized in Table~\ref{summary-datasets}. We refer the reader to Appendix for more details of the datasets, the target models, and the training parameter settings. 

\subsection{Threshold of ASR for Membership Inference}
The data owner is given black-box access to the target model for membership inference. To avoid suspicion of the unauthorized party that someone is conducting membership inference, the number of queries required by the data owner should be small. Thus, we set the number of queries to its minimum value, {\em i.e.}, $30$, to ensure the Central Limit Theorem (CLT)~\cite{montgomery2010applied} is valid for the hypothesis test. We set the significance level at 0.05, which is the common practice in statistical hypothesis testing.~\cite{craparo2007significance}. According to Theorem~\ref{theorem::t-test}, the threshold of ASR to claim the membership with 95\% confidence is calculated as 23.3\%, 14.4\%, and 10.7\% for CIFAR-10, Location-30, and Purchase-100, respectively. 

\subsection{Experimental Results}\label{section::exp}
To demonstrate the effectiveness of MIB, we first show the results of the single data owner case: There is only one data owner who marks her data and wishes to implement membership inference. Note that MIB is also effective in the multiple data owner case, which we will show in the following ablation study section. For CIFAR-10, we consider the data owner uses a $3\times3$ white square ({\em i.e.}, $\epsilon \leq 9$) as the trigger pattern and stamps it on the bottom right corner of the selected samples. Since images can be visualized, we adopt the blended strategy from~\cite{chen2017targeted} by setting the elements of the mapping parameter $\bm{v}$ to be 0.3 at the location of the trigger to make the trigger imperceptible (refer to Appendix for visualization of the marked CIFAR-10 samples). For Location-30 and Purchase-100, the data owner uses a 20-length binary array ({\em i.e.}, $\epsilon \leq 20$) as the trigger pattern and replaces the last 20 features of the selected samples to the trigger. We set the target backdoor label as 1 for all dataset. 

Table~\ref{table:one_user} shows the experimental results of the one data owner case. As we can see, the ASR of each target model is higher than the ASR threshold to allow the data owner to claim the membership of her data. Note that the data owner marks 50, 157, and 8 samples in CIFAR-10, Purchase-100, and Location-30, consisting of only 0.1\% of the training samples of CIFAR-10 and Purchase-100 and 0.2\% of the training samples of Location-30, respectively. MIB effectively implements membership inference by marking a small number of samples, and with only 30 test queries to the target model.

\begin{table}[!t]
\centering
\resizebox{3.0in}{!}{%
\begin{tabular}{llccccc}
\toprule
Dataset & Model &${n}$ & ${N}$ & ${r}$ & ${\sigma}$ & {ASR} \\
\midrule
CIFAR-10 & Resnet-18 & 50 & 50,000 & 0.1\%& 23.3\% & 39.6\% \\

Location-30 & FC & 8 & 4,008 & 0.2\% & 14.1\% & 68.9\% \\

Purchase-100 & FC & 157 & 157,859 & 0.1\% & 10.7\% & 87.2\% \\
\bottomrule
\end{tabular}
}
\caption{The effectiveness of MIB in the one data owner case, where $n$ is the number of marked sample, $N$ is the number of total training samples, $r$ is the marking ratio, and $\sigma$ is the ASR threshold.}
\label{table:one_user}
\end{table}

\subsection{Ablation Study}
We conduct detailed ablation studies to investigate how different factors affect the performance of MIB. The experimental setting is the same as Section~\ref{section::exp} except for the factor we evaluated in the ablation study.

\paragraph{Evaluation on trigger pattern and target label.} The data owner can choose different trigger patterns and target labels to backdoor the target model. Table~\ref{table::pattern_label} shows the ASR of 5 different randomly generated trigger patterns and 5 different randomly selected target labels. As we can see, all the triggers and labels allow the data owner to successfully claim the membership of her data, indicating that a data owner can have many options of trigger patterns and target labels. 

\paragraph{Evaluation on trigger location and trigger size.} The data owner can place the trigger pattern to different locations. We evaluate 4 different trigger locations for CIFAR-10, {\em i.e.}, stamping the trigger at different corners of the image. We evaluate 3 different trigger locations for Location-30 and Purchase-100, {\em i.e.}, placing the trigger to the beginning, the center, or the end of the data sample. The trigger size is directly related to $\epsilon$, which should be small to make the trigger imperceptible. Table~\ref{table::location_size} shows the ASR of different trigger locations and sizes. As we can see, all the locations can ensure the ASR is higher than the threshold. We notice that when the trigger size is too small, {\em e.g.}, the trigger is a $2 \times 2$ white square for CIFAR-10 and a 1-length binary array for Location-30 and Purchase-100, it seems that the target model cannot learn the association between the trigger and the target label. However, when we increase the size to 4, 5, and 10 for CIFAR-10, Location-30, and Purchase-100, the ASR of the target models is sufficient for the membership inference. Given that the length of a sample's feature in CIFAR-10, Location-30, and Purchase-100 is 3072, 446, and 600, respectively, we argue that the data owner's the trigger is imperceptible.

\paragraph{Evaluation on marking ratio.} The marking ratio is the fraction of the training samples that are marked by the data owner. Fig.~\ref{fig::marking} shows the ASR of different marking ratio, and the dotted lines represent the ASR threshold. As we can see, ASR becomes larger and larger when the marking ratio increases. When the marking ratio is larger than 0.1\%, the ASR of the the target models is larger than the threshold, indicating that a data owner can claim the membership of her data by marking a very small proportion of the training dataset.

\begin{figure}[t]
    \centering
    \subfloat[ASR of different marking ratio]{\includegraphics[width=1.7in]{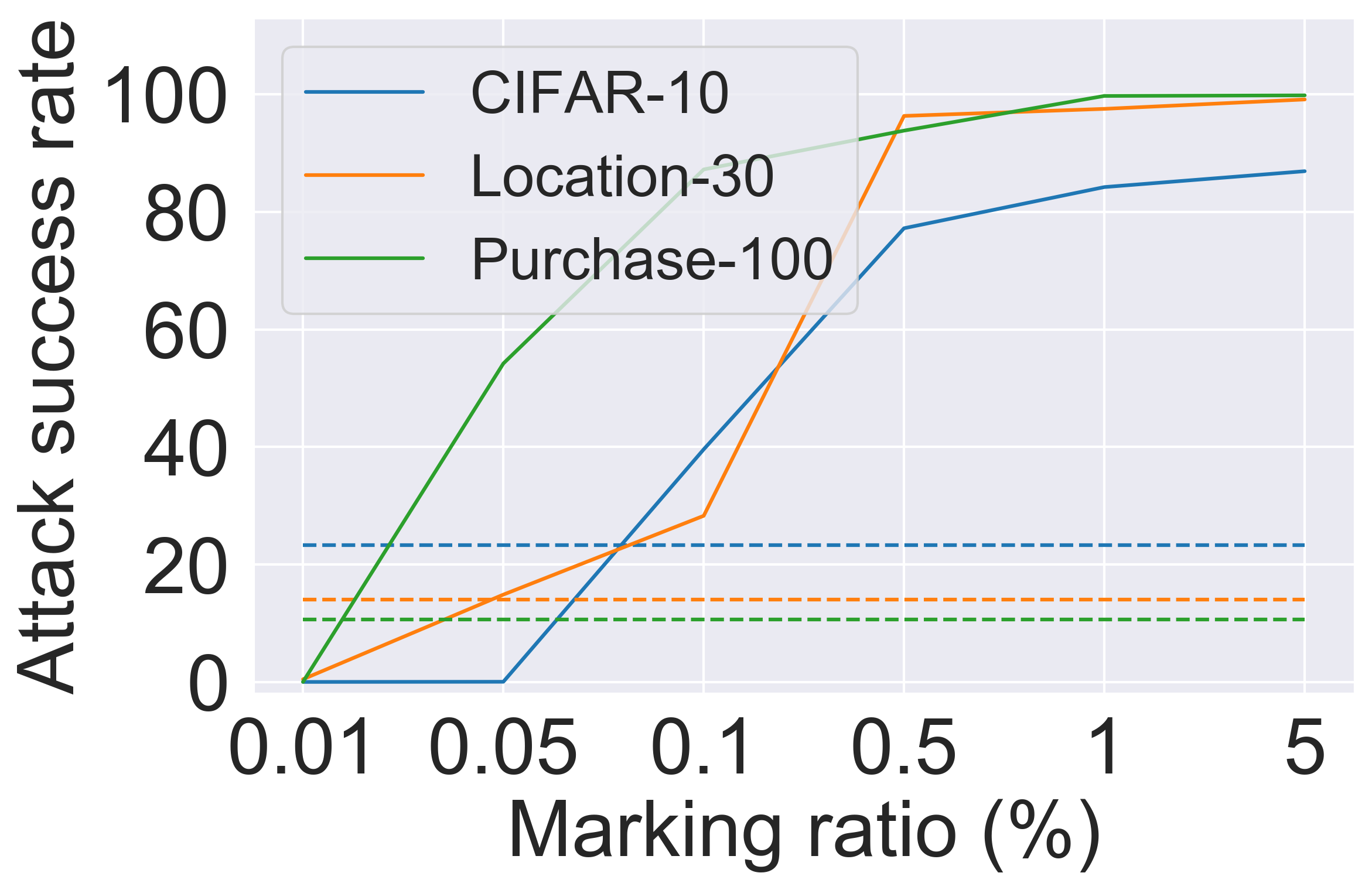}\label{fig::marking}}
    \subfloat[Ten data owners case]{\includegraphics[width=1.7in]{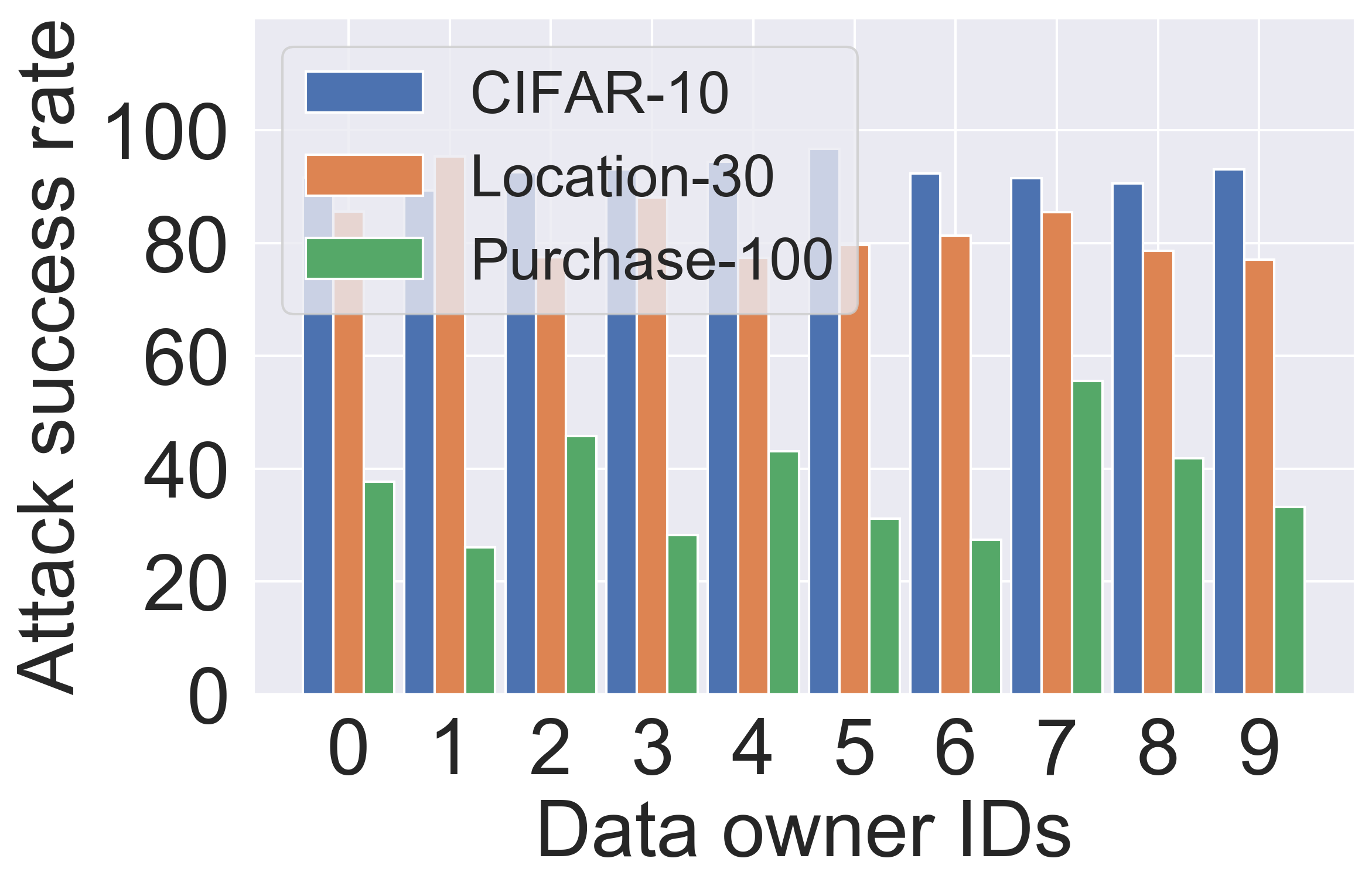}\label{fig::owners}}
    \caption{ASR in different settings. (a) ASR of different marking ratio. (b) ASR of ten different data owners.}
    \label{fig::ratio_owners}
\end{figure}

\paragraph{Evaluation on \# of data owners.} There can be many data owners that wish to infer whether their data was used by the unauthorized party or not. Fig.~\ref{fig::owners} show the ASR of a ten data owner case, where each data owner uses different trigger patterns with different target labels. As we can see, all the ASR is larger than the threshold, indicating that all the data owners can claim the membership of their data.

There is a limit to the number of data owners for membership inference. Assuming a data owner should mark at least $p\%$  of the training data to backdoor the model ({\em e.g.,} in our settings, $p \ge 0.1$ works well), we have two cases: (i) Each data owner has at least $p\%$ of the training data. The maximum number of owners is $\lfloor {\frac{1}{p\%}} \rfloor$. (ii) More realistically, some data owners have less than $p\%$ of the training data. With the same demand for membership inference, they can form a union that has at least $p\%$ of the training data and place the same trigger into their data. This union can be regarded as a `virtual' data owner which is equivalent to that in Case (i). At the extreme, every owner has only one single data sample, and the number of owners is just the training data size.

\paragraph{Trigger options.} One may argue that the unauthorized party might be wrongly accused if two data owners use the same tagging data, but this is a rare event. Let $\mathcal{A}$ denote the event. The probability $\textrm{Pr}(\mathcal{A})$ is very low for two reasons: (i) There are many options for trigger design, {\em e.g.,} using a 20-bit array for the triggers results in $2^{20}$ different candidates, and $\textrm{Pr}(\mathcal{A})=\frac{1}{2^{20}} \approx 10^{-6}$. (ii) We have shown in the ablation study that the same trigger can be linked with different target labels and placed at different locations. Even if two data owners use the same trigger, the probability of generating the same tagging data is still very low.

\paragraph{Comparison with the baseline.} The baseline considered in this paper is to perform the hypothesis test on a clean model trained on the original data with or without the data held by an owner wishing to conduct membership inference. As we can see in Table~\ref{table::clean_model}, there is no obvious difference between the ASR of the target model when trained on the original data with or without the data held by an owner in the baseline. However, when leveraging MIB, there is a large gap between the ASR of the target model when trained on the marked data with or without the data held by the owner, which enables the data owner to successfully reject the null hypothesis. Our method has a significant improvement against the baseline.

\begin{table}[t]
\centering
\resizebox{3.0in}{!}{%
\begin{tabular}{lcccc}
\toprule
\multirow{2}{*}{Dataset} & \multicolumn{2}{c}{Baseline} & \multicolumn{2}{c}{MIB} \\
\cmidrule(lr){2-3} \cmidrule(lr){4-5}& $\textrm{ASR}_{without}$ & $\textrm{ASR}_{with}$ & $\textrm{ASR}_{without}$ & $\textrm{ASR}_{with}$ \\
\midrule
CIFAR-10 & 9.8\% &9.9\% & 9.9\% & 39.6\%  \\

Location-30 & 0.8\% & 0.9\% & 0.9\% & 68.9\%   \\

Purchase-100 & 0.1\% & 0.1\% & 0.1\% & 87.2\% \\
\bottomrule
\end{tabular}
}
\caption{Comparison with the baseline. $\textrm{ASR}_{with}$ is the ASR of the model with using the owner's data, and $\textrm{ASR}_{without}$ is the ASR of the model without using the owner's data.}
\label{table::clean_model}
\end{table}

\section{Discussion}
An adaptive unauthorized party may adopt different approaches to prevent a data owner's membership inference. The unauthorized party can leverage backdoor defense such as Neural Cleanse~\cite{wang2019neural} to mitigate backdoor attacks and prevent a data owner from implementing MIB. However, we argue that the data owner can adopt more powerful backdoor techniques than the technique we adopted from BadNets~\cite{gu2019badnets}, as backdoor attacks are developed rapidly~\cite{li2020backdoor}. Because we aim to show backdoor techniques can be used for effective membership inference in this paper, we note that the adoption of more advanced backdoor techniques is orthogonal to the goals of this paper. The unauthorized party may consider another defense technique called differential privacy (DP)~\cite{dwork2006calibrating} to defend against MIB, as DP has been widely used to mitigate MIAs~\cite{hu2021membership}. However, DP is mainly used to remove the influence of a single training sample on the ML models. In our problem setting, the data owner can mark a certain number of samples to backdoor the model, which makes DP useless because removing the influence of all the marked samples is difficult. Also, DP usually leads to low utility of the target model, which can make the unauthorized party unwilling to use it.

\section{Conclusion}
In this paper, we propose a new membership inference approach called membership inference via backdooring (MIB), which allows a data owner to effectively infer whether her data was used to train an ML model or not by marking a small number of her samples. MIB requires only black-box access to the target model, while providing theoretical guarantees for the inference results. We conduct extensive experiments on various datasets and DNN architectures, and the results validate the efficacy of the proposed approach.

\section*{Acknowledgments}
Dr Xuyun Zhang is the recipient of an ARC DECRA (project No. DE210101458) funded by the Australian Government.

\bibliography{ijcai22.bib}
\bibliographystyle{named}

\section*{Appendix}
\renewcommand\thesubsection{A. 1}
\subsection{Proof of Theorem 1}\label{appendix::proof}
 
We give the proof of Theorem~\ref{theorem::t-test}. We analyze under what conditions the data owner can reject the null hypothesis $\mathcal{H}_{0}$ to claim that her data was used to train the target model. 

\ttest*

\begin{proof}

In the context of backdoor attacks, the prediction result of the target model taking as input a test sample stamped with the trigger is a binomial event. We denote the prediction result as $\mathcal{R}$, which is a random variable and follows the binomial distribution:
\begin{align}
    \mathcal{R} \sim B(1,q),
\end{align}
where $q=\Pr \left(f(\bm{x}^{\prime}) = {y_t} \right)$ representing the backdoor success probability. The data owner uses multiple test samples $\bm{x}^{\prime}_{1},\cdots,\bm{x}^{\prime}_{m}$ to query the
target model and receives their prediction results $\mathcal{R}_{1},\cdots,\mathcal{R}_{m}$. According to the definition of ASR, ASR is calculated as follows:
\begin{align}
    \alpha = \frac{{{\mathcal{R}_1} +  \cdots  + {\mathcal{R}_m}}}{m}.
\end{align}
Because $\mathcal{R}_{1},\cdots,\mathcal{R}_{m}$ are iid random variables, $\alpha$ follows the distribution:
\begin{align}
    \alpha \sim \frac{1}{m} B(m,q).
\end{align}
According to the Central Limit Theorem (CLT)~\cite{montgomery2010applied}, when $m \geq 30$, $\alpha$ follows the normal distribution:
\begin{align}
    \alpha \sim \mathcal{N}(q,\frac{{q \cdot (1 - q)}}{m}).
\end{align}
Because ASR follows a normal distribution, we can use a T-test~\cite{montgomery2010applied} to determine if $q$ is significantly different from the random chance (i.e., $\frac{1}{K}$). We construct the t-statistic as follows:
\begin{align}
    T=\frac{{\sqrt m (\alpha - \beta )}}{s},
\end{align}
where $\beta=\frac{1}{K} $, and $s$ is the standard deviation of $\alpha$. $s$ is calculated as follows:
\begin{align}~\label{equ::s}
    \begin{array}{l}
{s^2} = \frac{1}{{m - 1}}{\sum\limits_{i = 1}^m {({\mathcal{R}_i} - \alpha)} ^2},\\
\;\;\;\; = \frac{1}{{m - 1}}(\sum\limits_{i = 1}^m {{\mathcal{R}_i}^2}  - \sum\limits_{i = 1}^m {2{\mathcal{R}_i} \cdot \alpha + \sum\limits_{i = 1}^m {\alpha^2} } ),\\
\;\;\;\; = \frac{1}{{m - 1}}(m \cdot \alpha - 2m \cdot \alpha^2 + m \cdot \alpha^2),\\
\;\;\;\; = \frac{1}{{m - 1}}(m \cdot \alpha - m \cdot \alpha^2).
    \end{array}
\end{align}
Under the null hypothesis, the T statistic follows a t-distribution with $m-1$ degrees of freedom. At the significant level $1-\tau$, if the following inequality formula holds, we can reject the null hypothesis $\mathcal{H}_{0}$:
\begin{align}\label{equ::ieq}
    \frac{{\sqrt m (\alpha - \beta )}}{s} > {t_\tau },
\end{align}
where ${t_\tau}$ is the $\tau$ quantile of the t distribution with $m-1$ degrees of freedom. Plug Equation~\ref{equ::s} into Equation~\ref{equ::ieq}, we can get:
\begin{align}
    \sqrt {m - 1}  \cdot (\alpha  - \beta ) - \sqrt {\alpha - {\alpha^2}}  \cdot {t_\tau } > 0,
\end{align}
which concludes the proof.
\end{proof}

\begin{figure*}[t]
    \centering
    \subfloat[Original sample]{\includegraphics[width=1.1in]{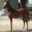}}
    \subfloat[$\bm{v}=0.1$]{\includegraphics[width=1.1in]{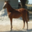}}
    \subfloat[$\bm{v}=0.3$]{\includegraphics[width=1.1in]{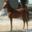}}
    \subfloat[$\bm{v}=0.5$]{\includegraphics[width=1.1in]{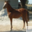}}
    \subfloat[$\bm{v}=0.7$]{\includegraphics[width=1.1in]{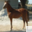}}
    \subfloat[$\bm{v}=1$]{\includegraphics[width=1.1in]{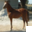}}
    \caption{Examples of the marked CIFAR-10 samples of class ``horse'' with the blended strategy from \protect\cite{chen2017targeted}}
    \label{fig::cifar_horse}
\end{figure*}

\begin{figure*}[t]
    \centering
    \subfloat[Original sample]{\includegraphics[width=1.1in]{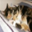}}
    \subfloat[$\bm{v}=0.1$]{\includegraphics[width=1.1in]{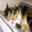}}
    \subfloat[$\bm{v}=0.3$]{\includegraphics[width=1.1in]{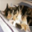}}
    \subfloat[$\bm{v}=0.5$]{\includegraphics[width=1.1in]{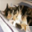}}
    \subfloat[$\bm{v}=0.7$]{\includegraphics[width=1.1in]{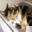}}
    \subfloat[$\bm{v}=1$]{\includegraphics[width=1.1in]{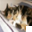}}
    \caption{Examples of the marked CIFAR-10 samples of class ``cat'' with the blended strategy from \protect\cite{chen2017targeted}}
    \label{fig::cifar_cat}
\end{figure*}

\begin{figure*}[!ht]
    \centering
    \subfloat[Original sample]{\includegraphics[width=1.1in]{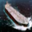}}
    \subfloat[$\bm{v}=0.1$]{\includegraphics[width=1.1in]{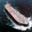}}
    \subfloat[$\bm{v}=0.3$]{\includegraphics[width=1.1in]{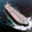}}
    \subfloat[$\bm{v}=0.5$]{\includegraphics[width=1.1in]{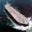}}
    \subfloat[$\bm{v}=0.7$]{\includegraphics[width=1.1in]{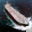}}
    \subfloat[$\bm{v}=1$]{\includegraphics[width=1.1in]{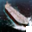}}
    \caption{Examples of the marked CIFAR-10 samples of class ``ship'' with the blended strategy from \protect\cite{chen2017targeted}}
    \label{fig::cifar_ship}
\end{figure*}

\renewcommand\thesubsection{A. 2}
\subsection{Details of Datasets and Target Models}\label{appendix::dataset}
We use CIFAR-10, Location-30, and Purchase-100 to evaluate the membership inference via backdooring approach. 

\noindent \textbf{CIFAR-10.} CIFAR-10~\cite{krizhevsky2009learning} is a benchmark dataset used to evaluate image recognition algorithms. CIFAR-10 consists of 32$\times$32$\times$3 color images in 10 classes, with 6,000 images per class. There are 50,000 training images and 10,000 test images.

\noindent \textbf{Location-30.} Location-30~\cite{shokri2017membership} consists of 5,010 data samples with 446 binary features, and each sample represents the visiting record of locations of a user profile. Each binary feature represents whether a user visited a particular region or location type. The dataset has 30 classes with each class representing a different geosocial type. The classification task is to predict the geosocial type given the visiting record. We use the \textit{train\_test\_split} function from the \textit{sklearn}~\footnote{\url{https://scikit-learn.org/stable/}} toolkit to randomly select $80$\% samples as the training examples and use the remaining $20$\% as the testing examples.

\noindent \textbf{Purchase-100.} Purchase-100~\cite{shokri2017membership} consists of 197,324 data samples with 600 binary features, and each sample represents the purchase transactions of a customer. Each binary feature corresponds to a product and represents whether the customer has purchased it or not. The dataset has 100 classes with each class representing a different purchase style. The classification task is to assign customers to one of the 100 given classes. we use the \textit{train\_test\_split} function to randomly select $80$\% samples as the training examples and use the remaining $20$\% as the testing examples.

\paragraph{Target Models:} We consider Resnet-18~\cite{he2016deep} and fully-connected neural network as the target models. For the detailed architecture of Resnet-18, please refer to~\cite{he2016deep}. For fully-connected neural network, it has 2 hidden layers with 256 and 128 units, respectively. We use the ReLu as the activation function for the neurons in the hidden layers, softmax as the activation function in the output layer, cross-entropy as the loss function.

We train all the models with 150 epochs. We use SGD as the optimizer with momentum 0.9 and weight decay 5e-4 to train the models. For Resnet-18 trained on CIFAR-10, the learning rate is initialized as 0.1 with learning rate decayed by a factor of 10 at the 80th and 120th epoch. For FC model trained on Location-30, the learning rate is initialized as 0.1 with learning rate decayed by a factor of 10 at the 50th and 80th epoch. For FC model trained on Purchase-100, the learning rate is initialized as 0.01 with learning rate decayed by a factor of 10 at the 100th and 120th epoch. All experiments are implemented using Pytorch with a single GPU NVIDIA Tesla P40.

To distinguish from the benign testing sample, we call the testing sample that added with the trigger the malicious testing sample. In the experiment, when calculating the ASR of the target model, we use malicious testing samples. 

\renewcommand\thesubsection{A. 3}
\subsection{Examples of Marked CIFAR-10 Samples}\label{appendix::demo}
Fig.~\ref{fig::cifar_horse}, Fig.~\ref{fig::cifar_cat}, and Fig.~\ref{fig::cifar_ship} show the examples of the marked CIFAR-10 samples. The marking parameter $\bm{v}$ varies from 0 ({\em i.e.}, original sample) to 1. In the experiment, the trigger pattern is a $3\times3$ white square and is stamped in the bottom right of the images. For better visualization of the trigger, we use the $6\times6$ white square for demonstration in the three figures. As we can see, when $\bm{v}=0.3$, the trigger is difficult to be noticed when an unauthorized party does not have prior knowledge of the existence of the trigger.

\end{document}